\documentclass[12pt,a4paper]{article}

\usepackage{amsmath,amsfonts, amssymb, graphicx}
\begin{document}
\bibliographystyle{abbrv}

\title{A  de Bruijn identity for symmetric stable laws}
\author{Oliver Johnson\thanks{Statistics Group, Department of Mathematics,
University of Bristol, University Walk, Bristol, BS8 1TW, UK.
Email {\tt o.johnson@bris.ac.uk}}
}
\date{\today}
\maketitle

\newtheorem{theorem}{Theorem}[section]
\newtheorem{lemma}[theorem]{Lemma}
\newtheorem{proposition}[theorem]{Proposition}
\newtheorem{corollary}[theorem]{Corollary}
\newtheorem{conjecture}[theorem]{Conjecture}
\newtheorem{definition}[theorem]{Definition}
\newtheorem{example}[theorem]{Example}
\newtheorem{condition}{Condition}
\newtheorem{main}{Theorem}
\newtheorem{remark}[theorem]{Remark}
\hfuzz30pt

\def \outlineby #1#2#3{\vbox{\hrule\hbox{\vrule\kern #1%
\vbox{\kern #2 #3\kern #2}\kern #1\vrule}\hrule}}%
\def \endbox {\outlineby{4pt}{4pt}{}}%
\newenvironment{proof}
{\noindent{\bf Proof\ }}{{\hfill \endbox
}\par\vskip2\parsep}
\newenvironment{pfof}[2]{\removelastskip\vspace{6pt}\noindent
 {\it Proof  #1.}~\rm#2}{\par\vspace{6pt}}
\newenvironment{skproof}
{\noindent{\bf Sketch Proof\ }}{{\hfill \endbox
}\par\vskip2\parsep}

\newcommand{\var}{{\rm{Var\;}}}
\newcommand{\cov}{{\rm{Cov\;}}}
\newcommand{\tends}{\rightarrow \infty}
\newcommand{\tz}{\rightarrow 0}
\newcommand{\C}[1]{{\bf ULC}(#1)}
\newcommand{\ep}{{\mathbb {E}}}
\newcommand{\pr}{{\mathbb {P}}}
\newcommand{\co}{{\mathbb {C}}}
\newcommand{\re}{{\mathbb {R}}}
\newcommand{\zz}{{\mathcal{Z}}}
\newcommand{\lle}{{\mathcal{L}}}
\newcommand{\snr}{{\textsf{snr}}}
\newcommand{\mmse}{{\textsf{mmse}}}
\newcommand{\roo}[2]{\rho^{M}_{#1,#2}}
\newcommand{\roof}[1]{\rho_{#1}^{F}}
\newcommand{\I}{\mathbb {I}}
\newcommand{\half}{\frac{1}{2}}
\newcommand{\bin}[2]{\binom{#1}{#2}}
\newcommand{\blah}[1]{}
\newcommand{\wh}[1]{\widehat{#1}}
\newcommand{\whopt}[1]{\widehat{#1}_{\rm opt}}

\begin{abstract} \noindent
We show how some attractive information--theoretic properties of Gaussians pass over to more 
general families of  stable densities. We define a new score function 
for symmetric stable laws, and use it to give a
stable version of the heat equation. Using this, we  derive a version of the de Bruijn identity, allowing us to write
the derivative of relative entropy as an inner product of score functions. We discuss maximum entropy properties of 
symmetric stable densities.
\end{abstract}

\section{Introduction and notation}
\subsection{Convergence of information--theoretic quantities}
A substantial body of literature (see for example \cite{artstein2, artstein, barron,bobkov6,brown,johnson14, linnik,
linnik2}) reformulates 
the classical Central Limit Theorem (CLT) in terms of 
  information--theoretic
quantities, such as Fisher information, entropy and relative entropy. We write $Z_{\sigma^2}$ for a Gaussian random 
variable with mean $0$, variance $\sigma^2$ and density $\phi_{\sigma^2}$. Given a probability density $f$ of variance $\sigma^2$, we write $H(f)$ and $D(f \| \phi_{\sigma^2})$ for the
entropy and relative entropy respectively.
 Recall the following  definition:
\begin{definition} \label{def:quantfish}
Fix a probability density $f$ of variance $\sigma^2$.  We write $J(f)$ for the standardized Fisher information
\begin{eqnarray} 
\label{eq:stanfisherdef} J(f) & = & \sigma^2 \int _{-\infty}^{\infty} f(x) \left( \roof{f}(x) +  \frac{x}{\sigma^2} \right)^2 dx.
\end{eqnarray}
 Here Fisher score function $\roof{f}(x) := \frac{f'(x)}{f(x)} =
 \frac{\partial}{\partial x} \log \left( f(x) \right)$, and we refer to 
\begin{equation} \label{eq:standscore} \roof{f}(x) +  \frac{x}{\sigma^2}   = \roof{f}(x) - \roof{\phi_{\sigma^2}}(x) = 
 \frac{\partial}{\partial x} \log \left( \frac{f(x)}{\phi_{\sigma^2}(x)}  \right) \end{equation}
 as the standardized Fisher score function, which vanishes if $f$ is Gaussian, and hence confirm $J(\phi_{\sigma^2}) = 0$.
Note that strictly speaking, $J(f)$ should be referred to as Fisher information with respect to location parameter, though we omit this
for brevity.
\end{definition}
Study of the Central Limit Theorem in this spirit
 began with  Linnik \cite{linnik,linnik2}, who used arguments
based on truncating and bounding densities (though note that \cite{barron} points out that Linnik's results must be regarded
as dubious, in that they contradict other known facts).  Interest in the information-theoretic approach to the 
CLT was revived and
developed in the 1980s by 
Brown \cite{brown} and Barron \cite{barron}. Writing $f_n$ for the density of an appropriately normalized sum of independent
identically distributed (IID)
random variables,
Brown \cite{brown} gave conditions for
 the standardized Fisher information $J(f_n)$ to converge to zero.
 Barron \cite{barron} proved necessary and sufficient conditions for $D(f_n \| \phi_{\sigma^2})$ to converge to zero,
or equivalently for the entropy $H(f_n)$ to converge to $H(\phi_{\sigma^2})$.
\subsection{de Bruijn identity}
 In fact,
Barron's work builds on Brown's, using the de Bruijn identity. This result
was first stated in differential form by Stam \cite[Equation (2.12)]{stam}, and proved in integral form  by Barron
\cite[Lemma 1]{barron}.  The de Bruijn identity considers $h_t$, the density of
$\sqrt{1-t} X + \sqrt{t} Z_{\sigma^2}$, which interpolates between a  given random variable $X$ with variance $\sigma^2$ and Gaussian $Z_{\sigma^2}$ with the same variance.
The de Bruijn identity can be understood in the context of the fact that $h_t$ satisfies a partial differential equation (PDE) of 
degree 2, the heat equation. As discussed in \cite[Example 2.5]{johnson34}, in this case we can state the heat equation in terms of
the standardized Fisher score function of Equation (\ref{eq:standscore}).
\begin{theorem}[Heat equation] \label{thm:heateqn}
Write $h_t$ for the density of $\sqrt{1-t} X + \sqrt{t} Z_{\sigma^2}$, where $X$ has variance $\sigma^2$. Then for $0 \leq t \leq 1$, density $h_t$ satisfies
\begin{equation} \label{eq:heat} \frac{ \partial h_t}{\partial t}(x) = \frac{ \sigma^2}{2(1-t)} \frac{ \partial}{\partial x}
\left( h_t(x) \left(\roof{h_t}(x) +  \frac{x}{\sigma^2} \right) \right). \end{equation}
\end{theorem}
Using integration by parts (see Section \ref{sec:stabdeb} for a more general version of the argument), Equation (\ref{eq:heat}) implies the following result:
\begin{theorem}[de Bruijn identity, differential form] \label{thm:debruijn}
Write $h_t$ for the density of $\sqrt{1-t} X + \sqrt{t} Z_{\sigma^2}$, where $X$ has variance $\sigma^2$. Then for $0 \leq t \leq 1$,
the relative entropy and standardized Fisher information are related by
\begin{equation} \label{eq:debruijn} \frac{ dD(h_t \| \phi_{\sigma^2})}{dt} = - \frac{1}{2(1-t)} J(h_t). \end{equation}
\end{theorem}
The main contributions of this paper are Theorems \ref{thm:pde} and \ref{thm:pdequant}, which extend Theorems \ref{thm:heateqn} and \ref{thm:debruijn} to the case of 
stable random variables.
\subsection{Maximum entropy and domains of normal attraction}
It is perhaps not a surprise that the relative entropy $D(f_n \| \phi_{\sigma^2})$ converges to zero in the Central Limit Theorem regime. It is well known (see for example \cite{johnson14} for a review) that under natural 
conditions, the entropy is maximised by probability densities of exponential family form. In particular (see \cite[Section 20]{shannon2})
if we fix the variance to be $\sigma^2$, the entropy is uniquely
maximised by the Gaussian density $\phi_{\sigma^2}$.

Of course, in the case of IID summands $X_i$, the standard CLT `square root' normalization does indeed fix the variance. In fact, assuming we normalize the sum of $n$ IID random variables by $\sqrt{n}$, finiteness of the variance of summands
is a necessary and sufficient condition for weak convergence to $\phi_{\sigma^2}$ (see \cite[Theorem 4, P181]{gnedenko}).
In general this assumption on the normalization defines the so-called
`domain of normal attraction':
\begin{definition} \label{def:dna} The domain of normal attraction of a stable law
is the set of $X$ such that taking $X_i$ are IID copies of $X$, then for some $A_n$ and $a$:
\begin{equation} \frac{ \sum_{i=1}^n X_i}{a n^{1/\alpha}} - A_n \mbox{ converges weakly to the stable law.} \label{eq:dna}
\end{equation} 
\end{definition}
 Hence, in some sense we understand the Central Limit Theorem as describing 
 convergence to a maximum entropy state. By combining these results of Shannon \cite{shannon2} and Gnedenko and Kolmogorov
\cite{gnedenko}  we state:
\begin{corollary} \label{cor:maxentdna}
The Gaussian density uniquely maximises entropy within its own domain of normal attraction. \end{corollary}
\subsection{Monotonicity and the Entropy Power Inequality}
This link suggests a stronger result; since $f_n$ converges weakly to $\phi_{\sigma^2}$ and $H(f_n) \leq H(\phi_{\sigma^2})$, 
it is natural to wonder whether $H(f_n)$ increases monotonically in $n$.
Indeed Barron writes in the Acknowledgement of  \cite{barron}  that
``[Professor Tom] Cover showed that Shannon's entropy power inequality implies the monotonicity of the entropy and
he posed the problem of identifying the limit.'' To expand on this, Shannon \cite[Theorem 15]{shannon2} stated the Entropy Power
Inequality (EPI), which gives a sharp bound on the entropy of the convolution of probability densities. The EPI was formally
proved by Stam \cite{stam} and Blachman \cite{blachman}, using arguments based on the de Bruijn identity, Theorem
\ref{thm:debruijn}, under the assumption of finite variance of the densities. We state the result under weaker assumptions, in a form due to 
Lieb \cite[Theorem 6]{lieb2}:
\begin{theorem}[Entropy Power Inequality] \label{thm:epi}
If,  for some
$p > 1$,  probability densities $f_X, f_Y \in L^p(dx)$ then
\begin{equation} 2^{2H( f_X \star f_Y )} \geq 2^{2H(f_X)} + 2^{2H(f_Y)}, \label{eq:epi} \end{equation}
with equality if and only if $f_X$ and $f_Y$ are Gaussian densities.
\end{theorem}
 Recent work \cite{bobkov9} clarifies further the conditions under which this result holds. Indeed, \cite[Corollary V6]{bobkov9}
shows that Equation (\ref{eq:epi}) holds if the entropies of $f_X$, $f_Y$ and $f_X \star f_Y$ exist. However
Lieb's result Theorem \ref{thm:epi} is sufficient for our purposes.

 In the context of the Central Limit Theorem, as Barron remarks, for IID variables the EPI implies that
along the `powers-of-two' subsequences $H( f_{2^k})$ and $D( f_{2^k} \| \phi_{\sigma^2})$ are monotone in $k$. Monotonicity
of the full entropy sequence $H(f_n)$, and equivalently of relative entropy, was only proved much later by Artstein et al.
\cite{artstein}, with later extensions from Madiman and Barron \cite{madiman}. This monotone increase of entropy 
suggests an interpretation of the Central Limit Theorem as a counterpart to the Second Law of Thermodynamics.
\subsection{Noisy communication channels}
An alternative perspective on all these results is given by Guo, Shamai and Verd\'{u} \cite{guo}, motivated
by communication through Gaussian channels. Instead of using the de
Bruijn identity (Theorem \ref{thm:debruijn}), they show that the derivative
of a certain mutual information quantity can be expressed in terms of  the minimum mean-squared error (MMSE) of the 
corresponding noisy channel. 
\begin{definition} \label{def:mmse}
For a noisy communication channel with input $X$ and output $Y$, we  measure the quality of a decoding rule
$\wh{X} = f(Y)$ in terms of the mean squared error $\ep( X - \wh{X})^2 = \ep(X - f(Y))^2$. It is well known that the optimal
decoding rule in this sense is of the form $\whopt{X} = \ep( X | Y)$, with corresponding MMSE equal to
\begin{equation} \label{eq:mmse}
\mmse( X | Y) := \ep \left( X - \whopt{X} \right)^2 = \ep \left( X - \ep(X | Y) \right)^2.
\end{equation}
\end{definition}
Using this definition, we now state \cite[Theorem 1]{guo}.
\begin{theorem}  \label{thm:mmse}
Consider input random variable $X$ (of finite variance) and output $Y$ linked by $Y = \sqrt{\snr} X + Z$, where
$Z$ is a standard Gaussian, and $\snr$ is a positive real parameter. Then
\begin{equation}
\frac{d}{d \snr} I(X; \sqrt{\snr} X + Z) = \frac{1}{2} \mmse( X | \sqrt{\snr} X + Z),
\end{equation}
where we write $I(U;V)$ for the mutual information between $U$ and $V$.
\end{theorem}
Verd\'{u} and Guo \cite{verdu} use these ideas to give an alternative proof of the Entropy Power Inequality, Theorem \ref{thm:epi}. 
\section{Summary of extensions to stable case}
However, note that all the theory described so far is tailored to the case of a Gaussian limit. It turns out that the 
Gaussian density has a number of attractive properties, which are  not easy to generalize to stable laws. 
In this case, as described for example in \cite[Chapter 7]{gnedenko}, there exists a fully
developed theory of necessary and sufficient conditions for domains of normal attraction for weak convergence, in the sense of (\ref{eq:dna}). However,
there are very few published results in the context of entropy and stable laws.

One reason that such results are elusive is that in general stable laws do not possess
moments of all orders, and most stable laws do not even have
densities that can be written down in closed form. The books by Samorodnitsky and Taqqu \cite{samorodnitsky} and Zolotarev
\cite{zolotarev3} review many of the relevant properties of stable laws.
For simplicity, we
will concentrate exclusively on the symmetric case where (in terms of the standard 
parameterization) $\beta = \mu = 0$ -- see Definition \ref{def:stabdens} below for a formal definition.
We now briefly summarise the main results and structure of this paper:
\begin{enumerate} \addtolength{\itemsep}{-5pt}

\item{ {\bf [Score properties, Section \ref{sec:scoredef}]}}  In Definition \ref{def:score},
we introduce a conditional expectation-based quantity,
which we call the MMSE score $\roo{X}{t}$.  Example \ref{ex:scoreGaussian} shows that
this MMSE score  $\roo{X}{t}(x)$ reduces to the standard Fisher score $\roof{h_t}$ in the 
Gaussian case ($\alpha = 2$).
Symmetric stable laws themselves have  score equal to $-x/s$ (see Lemma \ref{lem:tech})
which suggests
that $\roo{X}{t}(x)+x/s$ should be seen as the standardized MMSE score.

\item{ {\bf [Stable equivalent of the heat equation, Section \ref{sec:pde}]}} \label{prop:score}
For all symmetric stable laws
there exists a PDE  
in terms of $\roo{X}{t}$   (see Proposition \ref{thm:pde}).  This reduces to the heat equation, (\ref{eq:heat}), in the
Gaussian case ($\alpha = 2$).

\item{ {\bf [de Bruijn identity and channel derivative, Section \ref{sec:ders}]}}
Using this PDE, 
we deduce expressions for
the derivative of  (a) relative entropy -- see
Theorem \ref{thm:pdequant}, which generalizes the de Bruijn identity Theorem \ref{thm:debruijn}, and
(b) channel mutual information
-- see Theorem \ref{thm:verdu}, which generalizes Guo, Shamai and Verd\'{u}'s result, Theorem \ref{thm:mmse}.
\item{ {\bf [Maximum entropy and domains of normal attraction, Section \ref{sec:maxent}]}}
It is not the case that non-Gaussian stable laws are maximum
entropy within their own domains of normal attraction (see Lemma \ref{lem:notdoa}, where we prove that 
a counterpart of Corollary \ref{cor:maxentdna} does not hold in general).
 We can use the arguments described above to give a (not transparent) condition under which
the Cauchy is maximum entropy (see Lemma \ref{lem:maxent}).
\end{enumerate}

In Section \ref{sec:open} we conclude with some open problems, solution of which can help prove convergence in 
relative entropy to a symmetric stable law, using arguments in the spirit of Brown \cite{brown} and Barron \cite{barron}.

Note that Bobkov, Chistyakov and G\"{o}tze 
 consider entropy and stable laws
  in their paper  \cite{bobkov7}, which develops and extends the methods they introduced in
\cite{bobkov6}.  In some sense, their approach can perhaps be seen as a rigorous development of the ideas of Linnik
\cite{linnik,linnik2}, including proofs of bounds on the tails of characteristic functions, and hence of densities. In
particular, they do not provide maximum entropy results, or identities of de Bruijn type. However, their ideas are poweful 
enough to prove an optimal rate of convergence in the Central Limit Theorem regime in \cite{bobkov6}.

We first define the relevant stable laws:
\begin{definition}  \label{def:stabdens}
Given an exponent $\alpha \in (0,2]$, write $Z_s^{(\alpha)}$ for a centered $\alpha$-stable random variable with 
characteristic equation $\exp(- s|\theta|^{\alpha})$. In the standard parameterization, such a random variable has
scale 
parameter $c=s^{1/\alpha}$, shift parameter $\mu = 0$ and skewness
parameter $\beta = 0$. We write  that  $Z_s^{(\alpha)}$ has density 
\begin{equation} \label{eq:stabdens}  g_{s}^{(\alpha)}(x) = \frac{\zz}{s^{1/\alpha}} 
g \left( \frac{x}{s^{1/\alpha}} \right),\end{equation}
where the scaling constant $\zz$ is chosen such that $g(0) = 1$,
and the stability property implies that  the sum $Z_s^{(\alpha)} + Z_t^{(\alpha)} = Z_{s+t}^{(\alpha)}$,
or the convolution $g_s^{(\alpha)} \star g_t^{(\alpha)} = g_{s+t}^{(\alpha)}$. We write  $Z^{(\alpha)}$ for a standard variable
(i.e. with $s=1$) and $g^{(\alpha)}$ for its density.
\end{definition}
Since this notation may be slightly unfamiliar, we give two examples that fit in this framework, where we can be completely
explicit about the stable laws in question.
\begin{example} Using the notation of Definition \ref{def:stabdens}:
\begin{enumerate}
\item For $\alpha = 2$, taking $g(x) = \exp(-x^2/2)$ and
 $\zz = 1/\sqrt{2 \pi}$, we recover the Gaussian density with 
$s$ being the variance.
\item For $\alpha = 1$, taking $g(x) = 1/(1+x^2)$ and 
$\zz = 1/\pi$, we recover the Cauchy density.
\end{enumerate}
\end{example}
\section{Score function definition} \label{sec:scoredef}
Given a particular random variable $X$, we will suppose that we have reason to compare it with a particular stable random
variable $Z_s^{(\alpha)}$ of the form of Definition \ref{def:stabdens}. For example, we may suppose that $X$ is in the domain of
normal attraction of $Z_s^{(\alpha)}$.  We use a moment matching argument; in the case where $\alpha > 1$, since $\ep Z^{(\alpha)}_s = 0$, we use it to
approximate random variables $X$ with $\ep X = 0$.

We now define the MMSE score (with respect to $Z_s^{(\alpha)}$), denoted by $\roo{X}{t}(x)$, which is one of the main tools used in this paper:
\begin{definition} \label{def:score} 
Given a random variable $X$ with density $f$, we write $X_t = (1-t)^{1/\alpha}X + 
t^{1/\alpha} Z^{(\alpha)}_s$ for a sequence of random variables
which interpolate between $X$ and $Z^{(\alpha)}_s$. 
We write 
$f_t$ for the density of $(1-t)^{1/\alpha} X$, and $h_t$ for
the density of $X_t$.

For each $0 \leq t \leq 1$,
we define the  MMSE score function of $X$  as 
\begin{equation} \label{eq:scoredef}
 \roo{X}{t}(x) = - \frac{\ep \left[ t^{1/\alpha} Z^{(\alpha)}_s |  X_t = x \right]}{s t}
= - \frac{ \int_{-\infty}^{\infty} f_t(x-y) y g^{(\alpha)}_{s t}(y) dy}
{s t h_t(x)}.\end{equation}
\end{definition}
Observe that there is a clear MMSE interpretation to this score, which can be expressed in terms of an optimal
estimator, in the spirit of Definition \ref{def:mmse} and \cite{guo}. We explore this further in Section \ref{sec:mutinf} below.
Further, this MMSE score reduces to the Fisher score function of Definition \ref{def:quantfish} in the Gaussian case:
\begin{example} \label{ex:scoreGaussian}
For $\alpha= 2$, the $g_s^{(\alpha)}$ is Gaussian, and since
$y g_{st}^{(\alpha)} = y \exp(-y^2/(2t s)) = - s t \frac{d}{dy} \exp(-y^2/(2s t))$, then using integration by parts:
\begin{eqnarray} \int_{-\infty}^{\infty} f_t(x-y) y g^{(\alpha)}_{st}(y) dy
& = & s t \int_{-\infty}^{\infty} \frac{\partial f_t}{\partial y}(x-y) 
g^{(\alpha)}_{s t}(y) dy \nonumber \\
& = & - s t \int_{-\infty}^{\infty} \frac{\partial f_t}
{\partial x}(x-y) g^{(\alpha)}_{st}(y) dy = -s t  h'_t(x),
\end{eqnarray}
 so for each $t$, the
$\roo{X}{t} = \roof{h_t}$, the Fisher score of  Definition \ref{def:quantfish}.
\end{example}
We now show that symmetric stable laws, as in Definition \ref{def:stabdens}, have linear score:
\begin{lemma} \label{lem:tech}
For any $u$ and $v$, if $g^{(\alpha)}_s$ is of the form of Definition \ref{def:stabdens}
then
\begin{equation} \label{eq:condexp}
 \int_{-\infty}^{\infty} g^{(\alpha)}_u(x-y) g^{(\alpha)}_v(y) y dy = \frac{ v x}{u+v} g^{(\alpha)}_{u+v}(x)
\mbox{\;\;\;\; for all $x$.} \end{equation}
This means that for $X \sim g_s^{(\alpha)}$ itself stable, the MMSE score is $\roo{X}{t}(x) =-x/s$, or
equivalently the standardized MMSE score $\roo{X}{t}(x) + x/s$ vanishes.
\end{lemma}
\begin{proof}
The $g^{(\alpha)}_u(y)$ has characteristic function $
\int g^{(\alpha)}_u(y) \exp(i \theta y) dy = \exp(-u |\theta|^{\alpha})$,
so that $y g^{(\alpha)}_v(y)$ has characteristic function $\frac{1}{i}
\frac{\partial}{\partial \theta} \exp(- v |\theta|^{\alpha}) =
\frac{\alpha v}{i} \theta^{\alpha-1} \exp(-v |\theta|^{\alpha})$ for 
$\theta \neq 0$.

Hence the convolution of $g_u^{(\alpha)}(y)$ and $g^{(\alpha)}_v(y) y$ has characteristic
function equal to the product of the two expressions, that is
$$ \exp(-u |\theta|^{\alpha})
\frac{\alpha v}{i} \theta^{\alpha-1} \exp(-v |\theta|^{\alpha})
= \frac{v}{u+v} 
\left( \frac{\alpha (u+v)}{i} \theta^{\alpha-1} \exp(-(u+v) |\theta|^{\alpha}) \right),$$
which we recognise as $v/(u+v)$ times the characteristic function of
$x g^{(\alpha)}_{u+v}(x)$.

The formula for the MMSE score for $X \sim g^{(\alpha)}_s$ follows, since then $f_t = g^{(\alpha)}_{s(1-t)}$ and $h_t = g^{(\alpha)}_s$, and 
so we can apply (\ref{eq:condexp}) with $u= s (1-t)$ and $v = s t$ to decide that the denominator of (\ref{eq:scoredef}) is
$t x g^{(\alpha)}_s(x)$, and so (\ref{eq:scoredef}) becomes $-x/s$.
\end{proof}
We can give some explicit examples, where Lemma \ref{lem:tech} can be understood without using  
characteristic functions:
\begin{example}
In the Gaussian case ($\alpha = 2$), completing the square and writing $\gamma = u v/(u+v)$, we see 
that the LHS of Equation (\ref{eq:condexp}) is
\begin{eqnarray*}
\lefteqn{ \frac{1}{\sqrt{ (2 \pi)^2 u v}} \int_{-\infty}^{\infty}
\exp \left( - \frac{(x-y)^2}{2u} - \frac{y^2}{2v} \right) y dy  } \\
& = & \frac{1}{\sqrt{2 \pi (u+v)}} \exp \left(- \frac{x^2}{2(u+v)} \right)
\int_{-\infty}^{\infty} \frac{1}{\sqrt{2 \pi \gamma}}
\exp \left( - \frac{(y - x \gamma/u)^2}{2 \gamma} \right)
\left[ (y - x \gamma/u) + x \gamma/u \right] dy \\
& = & g^{(\alpha)}_{u+v}(x)
\int_{-\infty}^{\infty} \frac{1}{\sqrt{2 \pi \gamma}}
\exp \left( - \frac{(y - x \gamma/u)^2}{2 \gamma} \right)
\left[ (y - x \gamma/u) + x \gamma/u \right] dy \\
& = & g^{(\alpha)}_{u+v}(x) \left( 0 + x \gamma/u \right),
\end{eqnarray*}
and the result follows.
\end{example}

\begin{example} In the Cauchy case ($\alpha = 1$), we can use a partial
fraction argument based on that
of \cite{blyth}. That is, for some $A,B,C,D$ (which are functions of $x$
but not $y$),
we know that the LHS  of Equation (\ref{eq:condexp})  is
\begin{eqnarray*}
\lefteqn{
\frac{1}{\pi^2} \int_{-\infty}^{\infty} \frac{u}{u^2 + (x-y)^2}
\frac{v y}{v^2 + y^2} dy } \\
& = & \frac{1}{\pi^2} \int_{-\infty}^{\infty} \frac{A(x-y) + B}{u^2 + (x-y)^2}
+ \frac{C y + D}{v^2 + y^2} dy \\
& = & \frac{1}{\pi} \left( \frac{B}{u} + \frac{D}{v} \right).
\end{eqnarray*}
Now, as in \cite{blyth}, equating coefficients of $y$ in the equation
$u v y = (A(x-y) + B)(v^2 + y^2) + (Cy+D)(u^2 + (x-y)^2)$ shows
that 
$$B/u + D/v = v x/((u+v)^2 + x^2),$$
and the result follows.
\end{example}

Observe that the standardized MMSE score 
$\roo{X}{t}(x) + x/s$ has mean zero in the case where $\alpha \geq 1$
\begin{eqnarray*}
\lefteqn{ 
 \int_{-\infty}^{\infty}h_t(x) \left( \roo{X}{t}(x) +  \frac{x}{s}  \right) dx } \\
& = & - \frac{(1-t)}{s t} \int_{-\infty}^{\infty}f_t(x-y) g^{(\alpha)}_{st}(y) y dy dx + \frac{1}{s}
\int_{-\infty}^{\infty}f_t(x-y) (x-y) g^{(\alpha)}_{st}(y) dy dx \\
& = & - \frac{(1-t) t^{1/\alpha}}{s t} 
\ep Z_s^{(\alpha)} + \frac{ (1-t)^{1/\alpha}}{s} \ep X.\end{eqnarray*}
Hence if $\alpha > 1$ then by assumption $\ep Z_s^{(\alpha)} = \ep X = 0$, and
the mean is zero. If $\alpha = 1$, this becomes $(1-t)  \ep(Z^{(\alpha)}_s-X)$,
which can be assumed to be zero by (pseudo)-moment matching.

\section{Partial differential equation for $h_t$} \label{sec:pde}

The next result we prove is a partial differential equation in terms of
$t$ and $x$, involving the standardized MMSE score. It can be seen that this is a generalization of the heat equation,
Theorem \ref{thm:heateqn}:
\begin{theorem} \label{thm:pde}
The density $h_t$ of $X_t = (1-t)^{1/\alpha} X + 
t^{1/\alpha} Z^{(\alpha)}_s$ satisfies
\begin{equation} \label{eq:pde} \frac{\partial h_t}{\partial t} (x)
= \frac{s}{\alpha (1-t) } \frac{\partial}{\partial x} 
\biggl( h_t(x) \left( \roo{X}{t}(x) + \frac{ x}{s} \right) \biggr), \end{equation}
where $\roo{X}{t}$ is the MMSE score of Definition \ref{def:score}.
\end{theorem}
\begin{proof}
The key is to observe that each of $f_t$ and $g_{st}^{(\alpha)}$ 
satisfy partial differential equations which can be combined
together. Specifically, since $f_t(z) = f(z/(1-t)^{1/\alpha})/(1-t)^{1/\alpha}$,
we know that
\begin{eqnarray*}
 \frac{\partial f_t}{\partial t} (x-y)
& =  & \frac{1}{\alpha (1-t)} \left( f_t(x-y) + (x-y) \frac{\partial
f_t}{\partial x} (x-y) \right) \nonumber \\
& =  & \frac{1}{\alpha (1-t)} \left( f_t(x-y) + x \frac{\partial
f_t}{\partial x} (x-y) \right) - \frac{1}{\alpha (1-t)} y \frac{\partial
f_t}{\partial x} (x-y) 
\end{eqnarray*}
Now, multiplying by $g^{(\alpha)}_{st}(y)$ and integrating, we obtain
\begin{eqnarray}
\lefteqn{
\int_{-\infty}^{\infty} \frac{\partial f_t}{\partial t} (x-y) g^{(\alpha)}_{st}(y) dy } \nonumber \\
& = & \frac{1}{\alpha (1-t)} \left( h_t(x) + x \frac{\partial
h_t}{\partial x} 
(x) \right) -  \frac{1}{\alpha (1-t)} \int_{-\infty}^{\infty} y \frac{\partial
f_t}{\partial x} (x-y) g^{(\alpha)}_{st}(y) dy \nonumber \\
& = & \frac{1}{\alpha (1-t)}  \frac{\partial}{\partial x} (x h_t(x))  -  \frac{1}{\alpha (1-t)} \int_{-\infty}^{\infty} y \frac{\partial
f_t}{\partial x} (x-y) g^{(\alpha)}_{st}(y) dy. \label{eq:fpde} \end{eqnarray}

Similarly,  $g^{(\alpha)}_{st}(z) = g^{(\alpha)}_s(z/t^{1/\alpha})/t^{1/\alpha}$, so that
\begin{eqnarray}
 \frac{\partial g^{(\alpha)}_{st}}{\partial t} (y)
= - \frac{1}{\alpha t} \left( g^{(\alpha)}_{st}(y) + y \frac{\partial g^{(\alpha)}_{st}}{\partial y} 
(y) \right)
= - \frac{1}{\alpha t} \frac{\partial}{\partial y} \left( y g^{(\alpha)}_{st}(y) \right). \label{eq:gpde} \end{eqnarray}
This means that, using integration by parts, and the fact that
$\frac{\partial f_t}{\partial y} (x-y) = 
- \frac{\partial f_t}{\partial x} (x-y)$, the term
\begin{eqnarray}
\int_{-\infty}^{\infty} 
f_t(x-y) \frac{\partial  g^{(\alpha)}_{st}}{\partial t}(y) & =& 
\frac{1}{\alpha t}  \int_{-\infty}^{\infty} \frac{\partial f_t}{\partial y}
(x-y) y g^{(\alpha)}_{st}(y)
dy \nonumber \\
 & =& - 
\frac{1}{\alpha t}  \left( \int_{-\infty}^{\infty} y \frac{\partial f_t}{\partial x}
(x-y)  g^{(\alpha)}_{st}(y) dy \right). \label{eq:gpde2}
\end{eqnarray}
Adding Equations (\ref{eq:fpde}) and (\ref{eq:gpde2}) we obtain
\begin{eqnarray*}
\frac{\partial}{\partial t} h_t(x) & = & 
\frac{1}{\alpha (1-t)}  \frac{\partial}{\partial x} (x h_t(x))
- \frac{1}{\alpha t(1-t)}  \left( \int_{-\infty}^{\infty} y \frac{\partial f_t}{\partial x}
(x-y)  g^{(\alpha)}_{st}(y) dy \right) \nonumber \\
& = & \frac{1}{\alpha (1-t)}  \frac{\partial}{\partial x} (x h_t(x))
+ \frac{1}{\alpha (1-t)}  \frac{\partial}{\partial x} \left( -\frac{1}{t} \int_{-\infty}^{\infty} y  f_t(x-y)  g^{(\alpha)}_{st}(y) dy \right) \\
& = & \frac{1}{\alpha (1-t)}  \frac{\partial}{\partial x} (x h_t(x))
+ \frac{1}{\alpha (1-t)}  \frac{\partial}{\partial x} \left( s h_t(x) \roo{X}{t}(x)  \right)
\end{eqnarray*}
and the result follows.
\end{proof}
Note that Equation (\ref{eq:pde}) takes a particularly simple form, with low degree. In \cite[Section 5.3]{johnson14},  results of 
Medgyessy \cite{medgyessy2,medgyessy} were reviewed,
implying that certain stable densities satisfy PDEs. However, it should be noted
that the form of the PDEs depends on the form of any rational representation of $\alpha = m/n$, and in general such PDEs can have
arbitrarily large degree, making them unhelpful to derive de Bruijn identities such as Theorem \ref{thm:debruijn}.

\section{Derivatives of information--theoretic quantities} \label{sec:ders}
\subsection{de Bruijn identity for stable random variables} \label{sec:stabdeb}

Using this PDE, Equation (\ref{eq:pde}),  we can consider derivatives
of the  relative entropy $D( h_t \| g^{(\alpha)})$, entropy $H(h_t)$ and 
 energy functional $\Lambda_s^{(\alpha)}(X_t) 
:= -\ep \log g_s^{(\alpha)}(X_t)$.  Recall we write  $\roof{f}(x)  = f'(x)/f(x)$
for the Fisher score of density $f$.

We now prove Theorem \ref{thm:pdequant}.\ref{eq:pderelent} which generalizes the de Bruijn identity (Theorem
\ref{thm:debruijn}).
We view the RHS of Equation (\ref{eq:relentder}) as an inner product of two types of standardized score, 
firstly the standardized MMSE score $\roo{X}{t}(x) + x/s$ introduced in Definition \ref{def:score}, and secondly the standardized
 Fisher score  $\roof{h_t}(x) - \roof{g^{(\alpha)}_s}(x)$.
Example \ref{ex:scoreGaussian} shows that in the Gaussian case $\alpha =2$, these two scores coincide, and we 
recover the familiar standardized Fisher information $J(h_t)$ as $s$ times the expectation of a perfect square, proving
Theorem \ref{thm:debruijn}.

In general, we might hope to control the inner product in (\ref{eq:relentder}) using Cauchy-Schwarz, since we expect
that both terms will be close to zero when $f$ is close to $g_s^{(\alpha)}$.
\begin{theorem} \label{thm:pdequant} Consider $h_t$ the density of $X_t = (1-t)^{1/\alpha} X + 
t^{1/\alpha} Z^{(\alpha)}_s$.
\begin{enumerate}
\item \label{eq:pderelent} The derivative of the relative 
entropy is 
\begin{eqnarray}
  \frac{\partial}{\partial t} D(h_t \| g^{(\alpha)}_s)  
& = & - \frac{s}{\alpha (1-t)} \int_{-\infty}^{\infty}h_t(x) \left( \roo{X}{t}(x) + \frac{ x}{s}  \right)
\left(  \roof{h_t}(x) - \roof{g^{(\alpha)}}(x) \right) dx.\;\;\;\;\;\;
\label{eq:relentder} \end{eqnarray}
\item \label{eq:pdeent} The derivative of the entropy itself
\begin{equation} \label{eq:entder}
 \frac{\partial}{\partial t} H(h_t) = 
\frac{s}{\alpha (1-t)} \int_{-\infty}^{\infty}h_t(x) \left( \roo{X}{t}(x) +  \frac{ x}{s}  \right)
\left(  \roof{h_t}(x) \right) dx.\end{equation}
\item \label{eq:pdeener} $\Lambda(X_t)=  - \int_{-\infty}^{\infty}  h_t(x)
\log g_1(x) dx $ has derivative equal to
\begin{eqnarray}
 \frac{\partial}{\partial t} \Lambda_s^{(\alpha)}(X_t)
& = & \frac{1}{\alpha (1-t) } \int_{-\infty}^{\infty}
 h_t(x) (\roo{X}{t}(x) +  \frac{ x}{s}  ) \left( \roof{g^{(\alpha)}_s}(x) \right) dx. 
\label{eq:todeal}
\end{eqnarray}
\end{enumerate}
\end{theorem}

\begin{proof}
We give the argument for the derivative of the relative entropy -- in other cases, a similar argument will work.
Since $g^{(\alpha)}_s(x)$ remains constant in $t$, we  write
\begin{eqnarray}
 \lefteqn{ \frac{\partial}{\partial t} D(h_t \| g^{(\alpha)}_s) } \nonumber \\
& = &    \int_{-\infty}^{\infty}\frac{\partial h_t}{\partial t}(x) \log \frac{ h_t(x)}{
g^{(\alpha)}_s(x)} dx +
\int_{-\infty}^{\infty}h_t(x)  \left( \frac{\partial h_t}{\partial t}(x) \frac{1}{h_t(x)} \right) dx \label{eq:step1} \\
& = &    \int_{-\infty}^{\infty}\frac{s}{\alpha (1-t)} \frac{\partial}{\partial x} \left( h_t(x) \left( \roo{X}{t}(x) + \frac{ x}{s}  \right) \right)
 \log \frac{ h_t(x)}{ g^{(\alpha)}_s(x)} dx \label{eq:step2}  \\
& = & - \frac{s}{\alpha (1-t)} \int_{-\infty}^{\infty}h_t(x) \left( \roo{X}{t}(x) +  \frac{ x}{s}   \right)
\left( \frac{\partial h_t}{\partial x}(x) \frac{ 1}{h_t(x)}
- \frac{\partial g^{(\alpha)}_s}{\partial x}(x) \frac{ 1}{g^{(\alpha)}_s(x)} \right) dx. \label{eq:step3}
 \end{eqnarray}
Here, the second term in (\ref{eq:step1}) simplifies to give zero, so that (\ref{eq:step2}) follows using the generalized heat
equation Theorem \ref{thm:pde}, and (\ref{eq:step3}) follows using integration by parts, assuming all functions are well-behaved at infinity.
\end{proof}
We can use similar arguments to give expressions for derivatives of other entropy-like functionals (such as R\'{e}nyi or Tsallis entropies),
since for any function  $\Theta$, the derivative 
\begin{equation} \frac{d}{dt}
 \int_{-\infty}^{\infty} \Theta( h_t(x) ) dx =  - \frac{s}{\alpha (1-t)}  \int_{-\infty}^{\infty} \Theta''(h_t(x)) h_t(x)^2
\left( \roo{X}{t}(x) +  \frac{ x}{s}  \right) \left(  \roof{h_t}(x) \right) dx
\end{equation}
is an inner product with respect to a non-standard weighting.
\subsection{Derivative of mutual information} \label{sec:mutinf}
We can reproduce and extend the steps concerning the mutual information
considered by Guo, Shamai and Verd\'{u} \cite{guo}. First we make explicit the link between our MMSE score
and estimation. Recall that we write $X_t = (1-t)^{1/\alpha} X + t^{1/\alpha} Z^{(\alpha)}_s$.
\begin{remark} \label{rem:diffmmse}
Notice that  we can rephrase Equation (\ref{eq:scoredef}) to read 
 $s t \roo{X}{t}(w) = - t^{1/\alpha} \widehat{Z}(w)$,
where $\widehat{Z}(w) = \ep(Z^{(\alpha)}_s | X_t = w)$. Analogously defining
$\widehat{X}(w) = \ep(X | X_t = w)$, we obtain that
\begin{equation} \label{eq:scoremse}
 (1-t)^{1/\alpha} \widehat{X}(w) = \ep( (1-t)^{1/\alpha} X | X_t =w)
= \ep(X_t - t^{1/\alpha} Z | X_t = w) = w + s t \roo{X}{t}(w). \end{equation} 
This confirms the obvious fact that writing $\wh{X}$ for $\wh{X}(X_t)$ and $\wh{Z}$ for
$\wh{Z}(X_t)$ we obtain
\begin{equation} \label{eq:scoremse2} (1-t)^{1/\alpha} (X - \wh{X}) = (1-t)^{1/\alpha} X - s t
\roo{X}{t}(X_t) - X_t = t^{1/\alpha} (\wh{Z} - Z^{(\alpha)}_s), \end{equation}
so MMSE on estimating $X$ and $Z^{(\alpha)}_s$ agree up to a known factor.
\end{remark}
This allows us to prove the following result, which can be 
seen as a generalization of Theorem \ref{thm:mmse}, resulting in
an inner product representation similar to  Theorem \ref{thm:pdequant}.\ref{eq:pderelent} (though note that the resulting score functions are not precisely the standardized ones).
\begin{theorem} \label{thm:verdu}
For any random variable $X$,
\begin{equation} \label{eq:channel}
\frac{\partial}{\partial t} I(X; (1-t)^{1/\alpha} X+ t^{1/\alpha} Z^{(\alpha)}_s)
= \frac{s}{\alpha(1-t)}  \int_{-\infty}^{\infty} h_t(x) \left( \roo{X}{t}(x) + \frac{x}{s t} \right) \left( \roof{h_t}(x)
 \right) dx. \end{equation}
\end{theorem}
\begin{proof}
We expand
\begin{eqnarray*}
\lefteqn{ I(X; (1-t)^{1/\alpha} X+ t^{1/\alpha} Z^{(\alpha)}_s)  } \\
& = & H ( (1-t)^{1/\alpha} X+ t^{1/\alpha} Z^{(\alpha)}_s) - 
H( (1-t)^{1/\alpha} X+ t^{1/\alpha} Z^{(\alpha)}_s | X) \\
& = & H(h_t) - H(t^{1/\alpha} Z^{(\alpha)}_s ) \\
& = & H(h_t) - \frac{\log t}{\alpha} - H(Z^{(\alpha)}_s ).
\end{eqnarray*}
Hence, differentiating, and using Theorem \ref{thm:pde}
and Equation (\ref{eq:entder}), we obtain
\begin{eqnarray*}
\lefteqn{
\frac{\partial}{\partial t} I(X; (1-t)^{1/\alpha} X+ t^{1/\alpha} Z^{(\alpha)}_s) } \\
& = & \frac{s}{\alpha (1-t)} \int_{-\infty}^{\infty} h_t(x) \left( \roo{X}{t}(x) +  \frac{x}{s}  \right)
\left( \roof{h_t}(x)  \right) dx
- \frac{1}{\alpha t} \\
& = & \frac{s}{\alpha (1-t)} \int_{-\infty}^{\infty}  h_t(x) \ \left( \roo{X}{t}(x) +  \frac{x}{s} +  \frac{x (1-t)}{s t} \right) 
\left( \roof{h_t}(x)  \right) dx 
\end{eqnarray*}
since integration by parts means that $\int_{-\infty}^{\infty}x \frac{\partial h_t}{\partial x}(x)  dx =
- \int_{-\infty}^{\infty}h_t(x) dx = -1$.
\end{proof}

In the Gaussian case $\alpha = 2$, $s=1$,
we recover the  MMSE characterization of \cite{guo} (Theorem \ref{thm:mmse})  on observing that 
in this case, Example \ref{ex:scoreGaussian} gives
\begin{equation} \label{eq:verdugauss} \roof{h_t}(y) = \roo{X}{t}(y) =  \frac{ \sqrt{1-t} \wh{X} - y}{t}, \end{equation}
where the second identity follows from Equation (\ref{eq:scoremse}). Combining  (\ref{eq:verdugauss}) with
(\ref{eq:scoremse}), 
we  write the RHS of Equation (\ref{eq:channel}) as
\begin{eqnarray}
\frac{1}{2(1-t)} \ep \left[ \frac{\sqrt{1-t} \wh{X}}{t} \left( \frac{ \sqrt{1-t} \wh{X} - X_t}{t} \right) \right]
& = & \frac{1}{2t^2} \ep( \wh{X}^2 - 1)  \label{eq:project} \\
& = & -\frac{1}{2t^2} \mmse(X|X_t). \label{eq:project2}
\end{eqnarray}
Here (\ref{eq:project}) uses the fact that $\wh{X} = \ep( X | X_t)$, so we know that $\ep [\wh{X} X_t] = \ep[ \ep( X | X_t) X_t]
= \ep[ \ep(X X_t) | X_t] = \ep(X X_t) = \sqrt{1-t}$.
Equation (\ref{eq:project2}) means that we can deduce
\begin{eqnarray*}
\frac{\partial}{\partial t} I(X; \sqrt{1-t} X+ \sqrt{t} Z_1)
& = & -\frac{1}{2 t^2} \mmse(X| X_t)).
\end{eqnarray*}
We recover the exact form of Theorem \ref{thm:mmse} by a change of variable
argument, noting that the channel in \cite{guo} uses $\snr = (1-t)/t$, so that
$t = 1/(\snr +1)$, and $\frac{\partial t}{\partial \snr} = - 1/(\snr +1)^2
= -t^2$. This means that
$$ \frac{\partial I}{\partial \snr} = 
\frac{\partial I}{\partial t} \frac{\partial t}{\partial \snr} 
= \frac{\mmse(X | \sqrt{\snr} X + Z)}{2},$$
and we recover Theorem \ref{thm:mmse}.
\section{Maximum entropy properties} \label{sec:maxent}
We now briefly discuss maximum entropy results for stable densities. Recall that
the Gaussian is maximum entropy in  a class (random variables with given variance) which is defined by tail (moment) behaviour alone,
and which coincides with the domain of normal attraction, in the sense of Definition \ref{def:dna}. We show that the position is more complicated for more general stable laws.
In particular, in Lemma \ref{lem:notdoa} we show that there is no maximum entropy characterization of all stable laws
within their domain of normal attraction, of the kind that Corollary \ref{cor:maxentdna} establishes for the Gaussian 
density. 

This suggests that extra conditions are necessary, in a way that is  reminiscent of the case of Poisson variables (see \cite{johnson21}). It is well
known that the Poisson is not maximum entropy on fixing the mean, since the geometric has larger entropy. However
in \cite{johnson21}, it was proved that Poisson random variables (with probability mass functions $\Pi_\lambda$)  are maximum entropy within the class of 
variables with fixed mean and mass function $f$ such that $f/\Pi_\lambda$ is (integer) log-concave. In other words, we require both a tail condition and
pointwise control.

In this paper we do not give a direct definition of a class within which stable laws are  maximum entropy . However,
Lemma \ref{lem:maxent} gives a condition which needs to be verified for all $t$. We hope to simplify this condition in future work.
For example, in a similar spirit to \cite{johnson21}, we might hope that stable densities $g^{(\alpha)}$  are maximum entropy
in the class of random variables with density $f$ such that $f/g^{(\alpha)}$ is log-concave.
\begin{lemma} \label{lem:notdoa}
Given any $\alpha$-stable random variable $Z^{(\alpha)}_s$, 
if $\alpha < 2$
there exists $X$ such that
\begin{enumerate}
\item $H(X) > H(Z^{(\alpha)}_s )$ ,
\item $X$ lies in in the domain of normal attraction of $Z^{(\alpha)}_s$.
\end{enumerate}
\end{lemma}
\begin{proof} Consider $U_i$ IID $\sim Z^{(\beta)}$
and  (independently) take $Z_i$ IID $\sim Z^{(\alpha)}_s$,  where $2 \geq 
\beta > \alpha$, and take
$X_i = U_i + Z_i$. Then 
\begin{enumerate}
\item Since $Z_i$ and $U_i$ have integrable characteristic functions, their densities are uniformly bounded by  bounded constants $\zz^{(\beta)}$
and $\zz^{(\alpha)}$. Hence their densities certainly lie in $L_p(dx)$ for any
$p > 1$, and so we can apply Lieb's form of the Entropy Power
Inequality, Theorem \ref{thm:epi}, to deduce that
$H(X_i) = H(Z_i+U_i) >  H(Z_i) = H(Z^{(\alpha)}_s)$.
Here, the strict inequality follows since the density of $U_i$ is bounded by $p_U \leq \zz^{(\alpha)}$, or $-\log p_U(x) \geq - \log \zz^{(\alpha)}$ for all $x$.
Multiplying by $p_U(x)$ and
integrating,  we deduce that $H(U_i) \geq  - \log \zz^{(\alpha)} > -\infty$, so $2^{2H(U_i)} > 0$, which gives $H(Z_i+U_i) >  H(Z_i) $ using the EPI.

\item Further,  for IID copies $X_i \sim X$, the
 $(X_1 + \ldots + X_n)/n^{1/\alpha}$ converges weakly to $Z^{(\alpha)}_s$, since
\begin{eqnarray*} \frac{X_1 + \ldots + X_n}{n^{1/\alpha} }
& = & \frac{1}{n^{1/\alpha-1/\beta}} \frac{U_1 + \ldots + U_n}
{n^{1/\beta}  }
+ \frac{Z_1 + \ldots + Z_n}{n^{1/\alpha} } \\
& \sim & \frac{1}{n^{1/\alpha-1/\beta}} Z^{(\beta)} + Z^{(\alpha)}_s,
\end{eqnarray*}
where the first term tends to zero since $1/\alpha - 1/\beta > 0$.
\end{enumerate}
\end{proof}

Lemma \ref{lem:notdoa}
 also tells us that no equivalent of monotonicity of entropy holds in general; that is, since $H(X) > H(Z^{(\alpha)}_s)$
for some $X$ in the domain of normal attraction, it cannot be the case that entropy  is always increasing on convolution for random variables in this set.

The derivative of energy given in Theorem \ref{thm:pdequant}.\ref{eq:pdeener}  can allow us to deduce a 
maximum entropy result, in certain circumstances. The strategy is similar to that in \cite{johnson21}.
That is, we hope to prove that  the energy functional $\Lambda_s^{(\alpha)}(X_t) = -\int_{-\infty}^{\infty}  h_t(x) 
\log g_s^{(\alpha)}(x) dx$ is increasing in $t$. If that is the case, then we know that $\Lambda_s^{(\alpha)}(X_0) \leq \Lambda_s^{(\alpha)}(X_1)$, so
 since by construction $h_0 = f$ and $h_1 = g_1^{(\alpha)}$: 
\begin{eqnarray}
 H(f) & = & \int_{-\infty}^{\infty}-h_0(x) \log h_0(x) dx \nonumber \\ 
& \leq & \int_{-\infty}^{\infty}-h_0(x) \log g_s^{(\alpha)}(x) dx 
\label{eq:gibbs} \\
& \leq & \int_{-\infty}^{\infty}-h_1(x) \log g_s^{(\alpha)}(x) dx = H(g_1^{(\alpha)}), \label{eq:maxent}
\end{eqnarray}
The first inequality 
(\ref{eq:gibbs}) follows
by the Gibbs inequality.
This allows us to give a (not at all explicit) condition for a class
among which the Cauchy is maximum entropy. 
\begin{lemma} \label{lem:maxent}
If for all $t$ and $x$, random variable $X$ has a MMSE score
such that $\roo{X}{t}(x) + x/s$ has opposite signs to $x$, then 
it has entropy less than that of the Cauchy.
\end{lemma}
\begin{proof}
Consider Equation (\ref{eq:todeal}),
which in the case of $\alpha =1$, becomes
\begin{equation} \label{eq:todocauch}
\frac{ \partial \Lambda_s^{(\alpha)}}{\partial t}(X_t) = - \frac{s}{(1-t) } \int_{-\infty}^{\infty}
 h_t(x) \left( \roo{X}{t}(x) +  \frac{x}{s} \right) \frac{2x}{s^2+x^2}  dx.\end{equation}
Hence, if $\roo{X}{t}(x) + x/s$ has the opposite sign to $x$ for all $x$ and $t$, then
the integrand in Equation (\ref{eq:todocauch}) is negative for each $x$, so overall, we deduce that
$\Lambda_s^{(\alpha)}(X_t)$ is increasing, and so the result follows by Equation (\ref{eq:maxent}).
\end{proof}
In fact this argument works in more generality:
Gawronski \cite{gawronski}
shows that stable laws are infinitely differentiable, with the $k$th
derivative having $k$ zeroes.  Hence,  $g_s^{(\alpha)}$ is unimodal and symmetric, with the 2nd 
term in Equation (\ref{eq:todeal}) having the opposite sign to $x$, as
required.
\section{Open problems} \label{sec:open}
We now briefly mention some open problems associated with the new MMSE score function of Definition \ref{def:score}.
Resolution of these would help significantly towards proving convergence in relative entropy to a stable law, in a
framework similar to that of \cite{barron}.
\begin{enumerate} \addtolength{\itemsep}{-5pt}
\item The analysis of the Fisher information by Brown in \cite{brown} is based on the fact that the Fisher score  $\rho^F$
satisfies a conditional expectation (projection) identity, a result which dates back to 
the work of Stam \cite{stam} and Blachman \cite{blachman}.
It would be of interest to prove a 
corresponding result for the MMSE score $\rho^M$ of Definition \ref{def:score}.
\item Such a projection identity could lead to a result corresponding to the subadditivity of Fisher information
on convolution (again see \cite{stam} and \cite{blachman}) -- allowing us to control the behaviour of the terms on the RHS of
(\ref{eq:relentder}).
\item Similarly, such a projection identity may allow us to control the sign of the standardized score $\roo{X}{t}$ in 
Lemma \ref{lem:maxent}, meaning that the maximum entropy property can be made more transparent for stable laws.
\item As mentioned previously, Barron \cite[Lemma 1]{barron} took the de Bruijn identity in differential form proved by 
Stam \cite[Equation (2.12)]{stam}, and extended it to give a representation of the relative entropy as an integral with respect to $t$, using 
an argument based on analytical properties of the relative entropy. It would be of interest to provide a similar 
representation of $D( f \| g^{(\alpha)}_s)$ as an integral, using (\ref{eq:relentder}).
\item It may be hoped that combining the subadditivity result and an integral form of the de Bruijn identity, then
convergence in relative entropy could be proved in the stable convergence regime of (\ref{eq:dna}).
\item Finally, it would be of interest to extend all this work to more general (non-symmetric) families of stable laws, 
removing restrictions on the parameterization made in Definition \ref{def:stabdens}.
\end{enumerate}

\end{document}